\documentclass[10pt, conference]{ieeeconf}      
\IEEEoverridecommandlockouts                              
\usepackage{cite}\usepackage{hyperref}
\usepackage{amsmath,amssymb,amsfonts}
\usepackage{algorithmic}
\usepackage{graphicx}
\usepackage{textcomp}
\usepackage{amsmath,amssymb,bm,bbm,mathrsfs,amscd}
\usepackage{calc}
\usepackage{color}
\usepackage{dsfont}
\usepackage{graphicx}
\usepackage{epstopdf}
\usepackage{epsfig}
\usepackage{tikz}
\usepackage{amsmath}
\usepackage{amsfonts}
\usepackage{amssymb}
\usepackage{rotating}
\usepackage{mathtools}
\usepackage{color}
\usepackage{subfig}
\usepackage{enumerate,pgfplots}

\newtheorem{proposition}{Proposition}
\newtheorem{lemma}{Lemma}

\newtheorem{remark}{Remark}

\newcommand{\ba}{\begin{array}}
\newcommand{\ea}{\end{array}}

\newcommand{\be}{\begin{equation}}
\newcommand{\ee}{\end{equation}}

\newcommand{\ds}{\displaystyle}

\newcommand{\mc}{\mathcal}

\newcommand{\ov}{\overline}

\def\1{\boldsymbol{1}}

\newcommand{\R}{\mathbb{R}}

\newcommand{\de}{\mathrm{d}}

\def\R{\mathbb{R}}

\tikzstyle{v_c}=[circle, draw,inner sep=2pt, minimum width=12pt, color=blue]
\tikzstyle{v_a}=[circle, draw,inner sep=2pt, minimum width=12pt, color=red]
\tikzstyle{edge} = [draw,thick,-,font=\small ]
\tikzstyle{label} = [draw,fill=black,font=\normalsize]

\def\BibTeX{{\rm B\kern-.05em{\sc i\kern-.025em b}\kern-.08em
	T\kern-.1667em\lower.7ex\hbox{E}\kern-.125emX}}
\def\MA#1{\textcolor{black}{#1}}	
	
	\title{\LARGE \bf Behavioral-feedback SIR epidemic model: analysis and control}
	\author{Martina~Alutto, Leonardo~Cianfanelli, 
Giacomo~Como, Fabio~Fagnani and Francesca~Parise
\thanks{Martina Alutto, Leonardo Cianfanelli, Giacomo Como and Fabio Fagnani are with the Department of Mathematical Sciences ``G.L.~Lagrange,'' Politecnico di Torino, 10129, Torino, Italy (e-mail: {\{martina.alutto; leonardo.cianfanelli; giacomo.como; fabio.fagnani\}@polito.it}). Giacomo Como is also with the Department of Automatic Control, Lund University, 22100 Lund, Sweden.
	Francesca Parise is with the School of Electrical and Computer Engineering, Cornell University, Ithaca, NY, USA (e-mail: fp264@cornell.edu).}
\thanks{This project was partially supported by the Italian Ministry of University and Research under the PRIN project ‘‘Extracting essential information and dynamics from complex networks’’, grant  no. 2022MBC2EZ and by the PCCW Affinito-Stewart Award.}
}

\begin{document}

\maketitle
\thispagestyle{empty}
\pagestyle{empty}

\begin{abstract}
	This paper investigates a behavioral-feedback SIR model in which the infection rate adapts dynamically based on the fractions of susceptible and infected individuals. We introduce an invariant of motion and we characterize the peak of infection. We further examine the system under a threshold constraint on the infection level. Based on this analysis, we formulate an optimal control problem to keep the infection curve below a healthcare capacity threshold while minimizing the economic cost. For this problem, we study a feasible strategy that involves applying the minimal necessary restrictions to meet the capacity constraint and characterize the corresponding cost.
\end{abstract}

\begin{keywords}
	Epidemic models, Susceptible-Infected-Recovered model, Optimal control problem.
\end{keywords}

\section{Introduction}\label{sec:introduction}
Predicting the spread of an infectious disease is essential to prevent healthcare systems from being overwhelmed, as seen during the COVID-19 pandemic. Public health policies rely on mathematical models to guide containment strategies, ensuring a balance between epidemiological control and socio-economic stability. Among these models, deterministic compartmental frameworks, such as the Susceptible-Infected-Recovered (SIR) model \cite{Kermack.McKendrick:1927, Hethcote2000TheMO}, offer fundamental insights into epidemic dynamics. 

Such fundamental models are typically derived under the assumption of a constant transmission rate, failing to account for behavioral adaptations that may arise in response to the epidemic. Yet, human responses to epidemics, such as social distancing or the adoption of face masks, significantly influence the spread of infectious diseases. To address this problem, compartmental models have been extended to incorporate social feedback mechanisms, in which the infection rate adapts dynamically according to the state of the epidemic. This class of models, known as feedback SIR models, introduces so-called \emph{behavioral functions}, often represented as smooth threshold, saturation, or polynomial functions, to capture mitigation efforts (see e.g. \cite{CAPASSO197843, korobeinikov2006lyapunov, franco2020feedback, Baker2020, Alutto2021OnSE}). 
The models in \cite{CAPASSO197843, korobeinikov2006lyapunov} considered a transmission rate that decreases with infection prevalence (e.g., due to increasing awareness), proving for example that the infection curve remains unimodal as in the classical SIR epidemic model. Further refinements have analyzed stability properties and derived closed-form solutions \cite{franco2020feedback, Baker2020, korobeinikov2006lyapunov}. The work in \cite{Alutto2021OnSE} extends this framework by considering an infection rate that depends on \textit{both fractions of susceptibles $x$ and infected $y$ individuals}. Building on this latter model, we investigate a behavioral-feedback SIR model with a specific functional form for the infection rate, satisfying the following monotonicity conditions: non-decreasing in $x$ and non-increasing in $y$. 
Intuitively, a large fraction of susceptibles leads to a higher infection rate due to lower perceived risk (a common scenario in the early phases of an epidemic) and a large fraction of infected leads to a reduced infection rate because of increased awareness. 
The main novelty compared to \cite{Alutto2021OnSE} is that we establish the existence of an invariant of motion, a fundamental tool for the analysis of epidemic models. As in the classical SIR model, this invariant allows us to derive a closed-form relationship between the fractions of susceptible and infected individuals throughout the epidemic evolution, directly linking them to the initial state. 

Additionally, the derived constant of motion allows us to provide a precise estimate of the maximum infection peak, which is essential for designing effective containment policies aimed at mitigating the epidemic impact on the public health system. For example, several works in the literature investigated lockdown strategies that minimize economic costs while preserving healthcare capacity \cite{zino2021analysis, acemoglu2021optimal, kruse2020optimal, Cianfanelli.ea:2021, Miclo.ea:2022, ACEMOGLU2024111391}. Notably, \cite{kruse2020optimal} formulates a finite-horizon control problem for a SIR model, proving that the optimal control is bang-bang with at most two switches. In \cite{Cianfanelli.ea:2021}, this framework is extended to infinite time horizons, demonstrating that stabilizing the infected fraction is more effective than stabilizing the reproduction number when taking into account the fact that the mortality of the disease is increasing in the fraction of infected individuals.
A key contribution in the literature is the filling-the-box strategy derived in \cite{Miclo.ea:2022}, which ensures that the infection curve remains below a threshold representing the healthcare system’s capacity, often referred to as the ICU constraint. The optimal strategy in this framework involves applying the minimum restrictions necessary to meet the threshold constraint, a principle also validated in \cite{ACEMOGLU2024111391}, which formulates an optimal control problem incorporating molecular and serology testing for epidemic mitigation.
Motivated by the analysis of the model dynamics under a threshold constraint, we formulate an optimal control problem similar to that in \cite{Miclo.ea:2022}, but with feedback incorporated into the SIR dynamics. We study the filling-the-box strategy for this problem and characterize its corresponding cost with a sensitivity analysis with respect to the initial conditions.

The paper is structured as follows: Section~\ref{sec:bf-sir} introduces the behavioral-feedback SIR epidemic model. Section~\ref{sec:threshold} provides some geometric considerations on the state space. Section~\ref{sec:opt} formulates the optimal control problem and studies the properties of the filling-the-box strategy. Section~\ref{sec:conclusion} discusses future research directions.

\section{Behavioral-feedback SIR epidemic model}\label{sec:bf-sir}
In this section, we introduce the specific behavioral-feedback SIR (BF-SIR) epidemic model considered in this paper and we investigate its dynamics. 
\subsection{Model definition}
We consider the following generalized SIR epidemic model, that describes the spread of a disease within a homogeneous population \textit{with a state dependent transmission rate},
\be\label{eq:bf-sir-model} \dot x = -\frac{b(x)}{1+ a y }xy,\quad	\dot y=\frac{b(x)}{1+ a y }xy-\gamma y, \quad \dot z=\gamma y\,,\ee
where $x=x(t)$, $y=y(t)$ $z = z(t)$ indicate the fractions of susceptible, infected and recovered individuals in the population, respectively. The function $b: [0,1] \to (0,+\infty)$ is non-decreasing, differentiable and captures how the transmission rate varies with the number of susceptibles: a higher fraction of susceptibles may lead to a lower perceived risk, encouraging social interactions and riskier behaviors that facilitate contagion. 
\MA{The constant $a \geq 0$ quantifies the sensitivity of the transmission rate to the prevalence of infection. A higher value of $a$ means that susceptibles significantly reduce their contact rate (e.g., by adopting protective behaviors or limiting social interactions) when the fraction of infected individuals increases, thus lowering the effective transmission rate, while $\gamma>0$ represents a constant recovery rate from the disease.}

Observe that the class of models in \eqref{eq:bf-sir-model} includes as a special case the classical SIR epidemic model first proposed by \cite{Kermack.McKendrick:1927}, where $a=0$ and $b$ is a positive constant. 
Moreover, our model generalizes \cite{franco2020feedback, Baker2020}, by adding a dependence on $x$. 
On the other hand, \eqref{eq:bf-sir-model} is a special case of the behavioral-feedback SIR epidemic model studied in \cite{Alutto2021OnSE}. The specific structure considered in this paper allows us to refine the analysis conducted in \cite{Alutto2021OnSE}, for example by explicitly deriving an invariant of motion and characterizing the peak of infection. 

The next result establishes the well-posedness of the model by showing that $$\Delta = \{ (x,y,z) \in \mathbb{R}^3_+:\, x+y+z= 1 \}$$ is positively invariant under \eqref{eq:bf-sir-model}, and provides some preliminary results on the dynamics. \MA{The proof is omitted, as it follows from standard arguments and can be found in \cite{thesisMartina}.}

\begin{lemma}\label{prop:wellposedness}
	For every initial state $(x(0), y(0),z(0)) \in \Delta$, the BF-SIR epidemic model \eqref{eq:bf-sir-model} admits a unique solution $(x(t), y(t),z(t)) \in \Delta$ for $t\ge0$. Moreover, 
	\begin{itemize}
		\item[(i)] $t \to x(t)$ is monotonically non-increasing for all $t\ge0$;  
		\item[(ii)] $y(t) >0$ for all $t\ge0$ if and only if $y(0)>0$; 
		\item[(iii)] there exists $x_{\infty} \in \mathbb{R}$ such that $$\lim_{t \to +\infty} (x(t), y(t),z(t)) = (x_{\infty}, 0, 1-x_{\infty}).$$
	\end{itemize}
\end{lemma}

To facilitate our analysis, we define the following sets
$$\mc S=\left\{(x,y)\in\R_{\geq 0}^2:\,x+y \leq 1\right\},$$
$$\mc S_+=\left\{(x,y)\in\R_{>0}^2:\,x+y \leq 1\right\},$$
and observe that there exists an invertible map $s : \mc S \to \Delta$ such that $s(x,y) = (x,y,1-x-y)$. For this reason, from now on, we shall omit the variable $z$ when denoting the state.

\subsection{Reproduction number}
As done in \cite{Alutto2021OnSE}, we can define a reproduction number for the BF-SIR epidemic model \eqref{eq:bf-sir-model} as 
\be\label{R} R(x,y) = \frac{xb(x)}{\gamma(1+a y)}\,.\ee
Note that the BF-SIR epidemic model \eqref{eq:bf-sir-model} can be rewritten in terms of the reproduction number \eqref{R} as
\be \label{eq:bf-sir-model-r}\dot x =- R(x,y) \gamma y, \quad \dot y= \left(R(x,y)-1\right) \gamma y, \quad  \dot z = \gamma y\,.\ee
The assumptions that $a\ge0$ and that $b(x)>0$ is differentiable and non-decreasing imply that
\begin{align}
	R_x(x,y)=& \frac{b(x)+x b'(x)}{\gamma (1+ay)}>0\,,	\label{partialR_x}\\
	R_y(x,y)=&\frac{-ax b(x)}{\gamma(1+a y)^2} \leq0 \,,	\label{partialR_y}
\end{align}
for every state $(x,y)$ in $\mc S$, namely, the reproduction number is increasing with respect to the fraction of susceptible individuals $x$ and non-increasing with respect to the fraction of infected individuals $y$.  

The unimodal behavior of the infection curve can be obtained as a special case of a more general result in \cite{Alutto2021OnSE}.
\begin{proposition}[Theorem 1,~\cite{Alutto2021OnSE}]\label{prop:peak}
	Consider the BF-SIR epidemic model \eqref{eq:bf-sir-model-r} with initial state $(x_0,y_0)$ in $\mc S$ such that $y_0>0$. Then, 
	\begin{enumerate}
		\item[(i)] if $R(x_0,y_0)\le1$, $t\mapsto y(t)$ strictly decreases for $t\ge0$;
		\item[(ii)]if $R(x_0,y_0)>1$, there exists a peak time $\hat t>0$ such that $t\mapsto y(t)$ strictly increases for $t$ in $[0,\hat t]$ and strictly decreases for $t\ge\hat t$.
	\end{enumerate}
\end{proposition}\medskip
For our analysis, it will prove convenient to define
\be\label{rho} \rho(x,y)= 1-\frac1{R(x,y)} = 1- \frac{\gamma(1+a y)}{xb(x)}\,,\forall(x,y) \in \mc S_+\ee
and observe that $\rho(x,y) \ge 0$ if and only if $R(x,y) \geq 1$. 

\subsection{Invariant of motion}
To have a better characterization of the dynamics, we next derive an invariant of motion for the BF-SIR epidemic model~\eqref{eq:bf-sir-model-r} in $\mc S_+$, that is, a function $\Psi: \mathcal{S}_+ \to \mathbb{R}$ such that $\Psi(x(t), y(t))$ remains constant along every solution of \eqref{eq:bf-sir-model-r}. Notice that, if $\Psi(x,y)$ is differentiable on $\mc S_+$, then the chain rule and \eqref{eq:bf-sir-model-r} imply that  $\Psi(x,y)$ is an invariant of motion for \eqref{eq:bf-sir-model-r} if and only if for every $(x,y) \in \mc S_+$
$$0=\dot\Psi= \Psi_x\dot x+ \Psi_y\dot y=\gamma y\left((R(x,y)-1)\Psi_y-R(x,y)\Psi_x\right),$$
or, equivalently, if and only if 
\be \label{Psi_partial}  \Psi_x(x,y) = \rho(x,y) \Psi_y(x,y) \,,\ee
for every $(x,y)$ in $\mc S_+$. 

\begin{proposition}\label{prop:invariant}
	For a constant $a\geq 0$ and a differentiable non-decreasing function $b:[0,1]\to(0,+\infty)$, let 
	\be \label{eq:g} g(x) = \int_{1}^{x} \frac{\gamma a}{s b(s)}\de  s \leq 0\,,\ee
	and 
	\be\label{psi}\psi(x,y)=y e^{-g(x)}-\int_{1}^{x}e^{-g(s)} \left(-1 + \frac{\gamma }{s b(s)}\right)\de s\,.\ee
	Then, $\psi(x,y)$ is differentiable in $\mc S_+$
	with 
	\be\label{eq:psi-partial}\psi_y(x,y)=e^{-g(x)}\,,\qquad \psi_x(x,y)=e^{-g(x)}\rho(x,y)\,. \ee
	Moreover, $\psi(x,y)$ is an invariant of motion for the BF-SIR epidemic model \eqref{eq:bf-sir-model-r} in $\mc S_+$. 
\end{proposition}
\begin{proof}
	The first equation in \eqref{eq:psi-partial} can be obtained by direct computation. To prove the second one, note that 
	$$\ba{rcl}\psi_x(x,y)
	&=&-y e^{-g(x)}g'(x)+e^{-g(x)}\left(1 -\frac{\gamma}{xb(x)}\right)\\
	&=&-y e^{-g(x)}\frac{\gamma a}{xb(x)}+e^{-g(x)}\left(1 -\frac{\gamma}{xb(x)}\right)\\
	&=&e^{-g(x)}\left(1-\frac{\gamma(1+ay)}{xb(x)}\right)\\
	&=&e^{-g(x)}\rho(x,y)\\
	&=&\psi_y(x,y)\rho(x,y)\,.\ea$$
	Since \eqref{Psi_partial} is satisfied, $\psi(x,y)$ is an invariant of motion. 
\end{proof}\smallskip

\subsection{Peak of infection}
For the classical SIR model, invariants of motion have been used to derive exact analytical solutions \cite{Harko2014ExactAS}. More generally, this approach has proven valuable in various epidemic models for determining both the final fraction of susceptible individuals and the peak fraction of infected individuals \cite{sontag2023explicit, feng2007final}. The next result characterizes the maximum infection peak of the BF-SIR model.

To establish the result, observe that, since $b(x)$ is positive, differentiable and non-decreasing, the function $f(x):=xb(x)$ is differentiable and strictly increasing in $[0,1]$, and thus invertible. Let us start with a technical lemma.
\begin{lemma}\label{lemma:tech}
	Suppose that there exists an initial condition $(x_0, y_0) \in \mc S_+$ such that $R(x_0, y_0)>1$. Define the value $\tilde x_{\gamma}=f^{-1}(\gamma)$ and the function $\tilde{x}$ as
	\begin{align}
		\tilde{x}(y) &:=f^{-1}(\gamma(1+a y)), \quad a > 0\,, \label{x_2}
	\end{align} 
	then $\tilde{x}_{\gamma} \in (0,1)$ and $\tilde{x}(y) \in (0,1]$ for all $y \in (0,\min \{1,\frac{f(1)-\gamma}{a \gamma}\}]$.
\end{lemma}\smallskip
\begin{proof}
	First observe that since there exists an initial condition $(x_0, y_0) \in \mc S_+$ such that $R(x_0, y_0)>1$, from \eqref{partialR_x} and \eqref{partialR_y}, we obtain $R(1,0)>1$ which means $f(1)>\gamma$. Hence, $\tilde{x}_{\gamma}$ is well defined. 
	From $f(\tilde{x}_{\gamma} )= \gamma$, $\gamma\in (0, f(1))$ and monotonicity of $f(x)$, it follows that $\tilde{x}_{\gamma} \in (0,1)$. 
	
	Next, note that since $a>0$, $\frac{f(1)-\gamma}{a \gamma}>0$ and $y \leq \frac{f(1)-\gamma}{a \gamma}$, it holds that $\gamma(1+ay) \leq f(1)$. Hence, $\tilde{x}(y)$ is well defined. From $f(\tilde{x}(y)) = \gamma(1+ay) \in (0, f(1)]$ and monotonicity of $f$, we obtain $\tilde{x}(y) \in (0,1]$.
\end{proof}\smallskip

Define the value $y_a$ as
\be\label{ya} y_a = \begin{cases}
	1 & \text{ if } a=0,\\
	\min \{1,\frac{f(1)-\gamma}{a \gamma}\} & \text{ if } a>0.
\end{cases}\ee

\begin{proposition}\label{prop:y-max}
	Consider the BF-SIR epidemic model \eqref{eq:bf-sir-model}. 
	Let $(x_0,y_0)$ be the initial state in $\mc S$ with $y_0>0$. Then, 
	\begin{enumerate}
		\item[(i)] if $R(x_0, y_0) \leq1$, then $y(t) \leq y_0$ for all $t\ge0$;
		\item[(ii)] if $R(x_0, y_0) >1$, define
		\be\label{phi}\phi(y) := \psi\left(f^{-1}(\gamma(1+a y)), y\right),\ee
		for all $y \in (0, y_a]$, then $\phi$ is invertible \MA{and there exists a finite time $\hat t$ at which $y(\hat t) = \phi^{-1}(\psi(x_0, y_0))$ and $y(t)\leq  \phi^{-1}(\psi(x_0, y_0))$ for all $t\ge0$.}
	\end{enumerate}
\end{proposition}\smallskip
\begin{proof}
	(i) Let $(x(t), y(t))$ be the solution of the BF-SIR epidemic model \eqref{eq:bf-sir-model}, with initial state $(x_0,y_0)$ in $\mc S$.
	If $R(x_0,y_0)\leq1$, from Proposition~\ref{prop:peak}(i), the fraction of infected $y(t)$ attains its maximum value at its initial state $y_0$. 
	
	(ii) Assume now that $R(x_0,y_0)>1$. From Proposition~\ref{prop:peak}(ii), there exists a peak time $\hat t>0$ such that $\dot y(\hat t) = 0$. From Lemma~\ref{prop:wellposedness}(ii), $y_0>0$ implies $y(t)>0$ for all $t \geq0$, then from \eqref{eq:bf-sir-model-r}, it follows that, $R(x(\hat t), y(\hat t))=1$, that is
	\be\label{x-hat} x(\hat t) \, b(x(\hat t))=\gamma(1+a y(\hat t))\,,\ee
	where $y(\hat t)>y_0$. It follows from \eqref{x-hat} that $y(\hat t)\in (0, y_a)$.
	Hence, by Lemma \ref{lemma:tech}, the fraction of susceptibles at the peak time $\hat t$ is 
	\be\label{x-hat2} x(\hat t) = f^{-1}\left(\gamma(1+a y(\hat t))\right) = \tilde{x}(y(\hat t))\,.\ee
	
	Moreover, note that $R(\tilde x(y),y) = 1$ for $y \in (0, y_a]$, which in turn is equivalent to
	\be\label{rho-hat} 
	\rho(\tilde x(y),y) = 0, \quad \forall y \in  (0, y_a]\,,\ee 
	by definition of $\rho$. From Proposition~\ref{prop:invariant}, it follows that $\psi(x(\hat t), y(\hat t)) = \psi(x_0,y_0)$. Moreover, \eqref{phi} and \eqref{x-hat2} imply 
	\be\label{psi-hat} \phi(y(\hat t))= \psi(\tilde x( y(\hat t)), y(\hat t))= \psi(x(\hat t), y(\hat t))= \psi(x_0,y_0).\ee 
	From \eqref{eq:psi-partial} and \eqref{rho-hat}, $\phi(y)$ is differentiable with derivative
	$$\mspace{-10mu}
	\ba{rcl}\label{phi-prime}
	\phi'(y)&=& \psi_x(\tilde x(y), y) \frac{a \gamma}{f'(\tilde x(y))}\!+\! \psi_y(\tilde x(y), y)\\
	&=& e^{-g(\tilde x(y))} \rho\left(\tilde x(y), y\right) \frac{a \gamma}{f'(\tilde x(y))}+ e^{-g(\tilde x(y))}\\
	&=& e^{-g(\tilde x(y))}\\
	&>& 0\,.\ea$$
	This means that $\phi$ is invertible and from \eqref{psi-hat} we get
	$$y(\hat t) = \phi^{-1}(\psi(x_0, y_0))\,.$$
	Proposition~\ref{prop:peak}(ii) implies that \MA{$\phi^{-1}(\psi(x_0, y_0))$ is the infection peak and thus $y(t) \leq \phi^{-1}(\psi(x_0, y_0))$ for all $t\ge0$}.
\end{proof}\smallskip
\begin{remark}
	Note that from Lemma~\ref{prop:wellposedness}(ii), if $y_0=0$ then the fraction of infected $y(t)$ will remain $0$ for all $t\ge0$. Then, for all $(x_0,y_0)$ in $\mc S$ we can define the peak of infection as 
	\be\label{y-max}
	\max_{t \geq 0} \, y(t)= \begin{cases}
		y_0 &\text{if } R(x_0, y_0) \leq1 \text{ or } y_0=0\\
		\phi^{-1}(\psi(x_0, y_0)) &\text{if } R(x_0, y_0) >1, y_0\neq 0.
	\end{cases}\ee
\end{remark}\medskip
\begin{remark}
	Understanding the maximum value of $y(t)$ is crucial, as usually social planners strive to keep the proportion of infected individuals below a predetermined threshold while minimizing the economic and social costs of restricting economic activity and social interactions. For example, this threshold may correspond to the healthcare system’s capacity to provide adequate care to infected patients. Adhering to this constraint ensures that all patients receive the necessary medical care, preventing the healthcare system from becoming overwhelmed.
\end{remark}

\section{Geometric considerations on the state space}\label{sec:threshold}
In this section, we provide some geometric considerations on the state space based on a threshold $\ov y$ in $(0,1)$ on the fraction of infected individuals, which will be interpreted as a constraint in the optimal control problem presented in Section~\ref{sec:opt}.

Given $\bar y$, define the portion of state space below the threshold as
\be \label{D}\mc D_{\ov y}=\{(x,y)\in\mc S: y \leq \ov y\}\,.\ee

We assume that the objective of the central planner is to keep the trajectory in $\mc D_{\ov y}$ so that the infected curve never exceeds the capacity $\ov y$. 
We start by deriving a simple sufficient condition on $\ov y$ such that the dynamics will never reach the capacity threshold, even without interventions.
\begin{lemma}
	If $R(1-\ov y, \ov y)\leq 1$, then $\mc D_{\ov y}$ is invariant for the BF-SIR epidemic model \eqref{eq:bf-sir-model}. 
\end{lemma}
\begin{proof}
	Notice that if $R(1-\ov y, \ov y)\leq 1$, we have that $R(x, y)< 1$ for every $y> \ov y$ and $x\in [0, 1-y]$, due to~\eqref{partialR_x}. Consequently, for every initial condition $(x_0, y_0)\in\mc D_{\ov y}$, the corresponding solution $(x(t), y(t))$ of \eqref{eq:bf-sir-model}, has a peak of infection \eqref{y-max} which is necessarily less or equal to $\ov y$ since, either $y(t) \leq y_0\leq \ov y$ or the peak must necessarily happen at a location where $R= 1$, by Proposition \ref{prop:peak}, and $R(x(t),y(t)) <1$ for all $y(t) > \ov y$. 
\end{proof}\smallskip

\begin{figure}
	\includegraphics[scale=0.6]{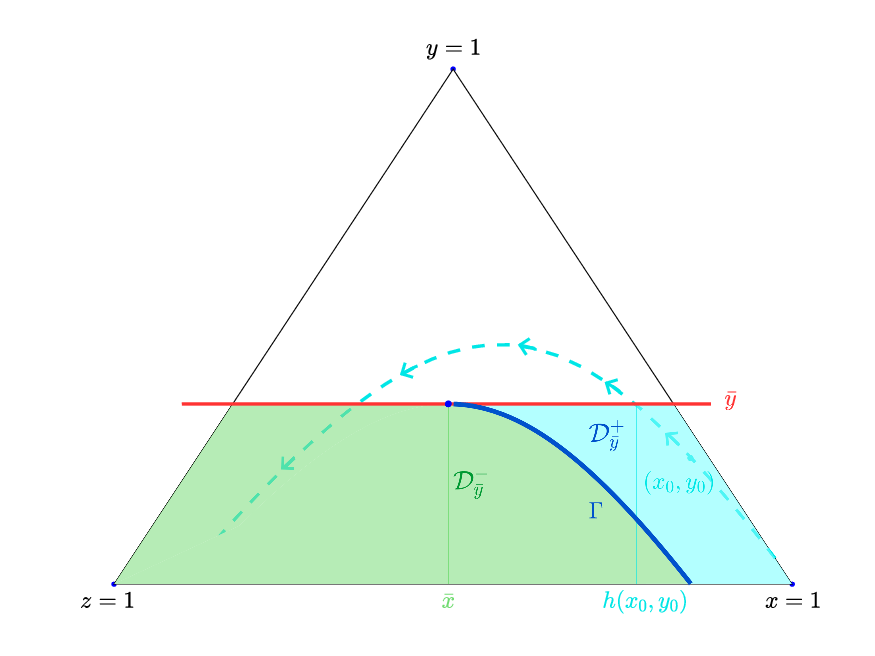}
	\caption{Phase portrait of the BF-SIR epidemic model \eqref{eq:bf-sir-model} on the simplex $\Delta$, with $b(x)=0.75x$, $a =1.5$ and $\gamma=0.05$. 
	}
	\label{fig:phase}
\end{figure}

Next, we show that if the condition in the previous lemma is not met, the state space can be divided in two regions, based on whether a trajectory that starts in $\mc D_{\ov y}$  would exceed the threshold, thus requiring interventions (see more details in Section \ref{sec:opt}), or not. 
Specifically, if $R(1-\ov y, \ov y)>1$, we aim at dividing $\mc D_{\ov y}$ in two regions:
\begin{enumerate}
	\item[(i)] $\mc D_{\ov y}^-$, such that if a solution starts therein, then $y(t) \leq \ov y$ for all $t \geq 0$,
	\item[(ii)] $\mc D_{\ov y}^+$, such that if a solution starts therein, then there exists a time instant $\hat t$ at which $y(\hat t) > \ov y$.
\end{enumerate}
To define these regions, we exploit the characterization of the peak of infection provided in \eqref{y-max}. More formally, let 
$$\Gamma = \{(x,y) \in \mc D_{\ov y}: R(x,y)> 1,\, \phi^{-1}(\psi(x,y)) = \ov y \}\,,$$
and 
\begin{align*}
	\mc D_{\ov y}^{-}= & \{(x,y)\in\mc D_{\ov y} : R(x,y)>1 ,\, \phi^{-1}(\psi(x,y))\!\leq  \ov y\} \, \cup\\[1pt]
	&\{(x,y) \in  \mc D_{\ov y} :R(x,y)\le  1 \text{ or } y =0\}\,,
\end{align*} 
$$ \mc D_{\ov y}^{+}\!=\!\!\{(x,y)\in\mc D_{\ov y} \!: R(x,y)\!>\!1, \phi^{-1}(\psi(x,y))\!>\! \ov y, y \!>\!0\}.$$
Note that 
$$ \mc D_{\ov y}^{-} \cup \mc D_{\ov y}^{+}  = \mc D_{\ov y}\,.$$
In Figure \ref{fig:phase}, the set $\Gamma$ is represented by the thick blue curve. The green and the light-blue areas are, respectively, the sets $\mc D_{\ov y}^{-}$ and $\mc D_{\ov y}^{+}$. 

We start by showing that $\mc D_{\ov y}^{-}$ is positively invariant for the BF-SIR epidemic model \eqref{eq:bf-sir-model}, which guarantees statement (i).
\begin{lemma}\label{lemma:invariant-D-}
	The set $\mc D_{\ov y}^-$ is positively invariant for the BF-SIR epidemic model \eqref{eq:bf-sir-model}. 
\end{lemma}
\begin{proof}
	Consider an initial state $(x_0, y_0) \in \mc D_{\ov y}^-$. Lemma~\ref{prop:wellposedness}(ii) ensures that if $y_0=0$, then $y(t)=0$ for all $t \geq 0$, thus the trajectory remains in $D_{\ov y}^-$. Hence, assume in the following that $y_0 >0$. 
	
	By definition of the set $D_{\ov y}^-$, we can distinguish two cases: either $R(x_0,y_0)>1$ with $\phi^{-1}(\psi(x_0,y_0))\leq  \ov y$ or $R(x_0,y_0)\le 1$. 
	First assume that $R(x_0,y_0)>1$ and $\phi^{-1}(\psi(x_0,y_0))\leq  \ov y$. Since $\psi(x,y)$ is an invariant of motion in  $\mc S_+$ by Proposition~\ref{prop:invariant}, $x_0>0$ and $y_0>0$, it follows that $\psi(x_0, y_0) = \psi(x(t), y(t))$ for all $t \geq 0$ and then, 
	$$\phi^{-1}(\psi(x(t),y(t))) = \phi^{-1}(\psi(x_0,y_0)) \leq \ov y\,,$$
	for all $t \geq 0$. Moreover, by \eqref{y-max}, $y(t) \leq \ov y$ for all $t \geq 0$.
	This means that the solution $(x(t), y(t))$ will remain in $\mc D_{\ov y}^-$ for  all $t \geq 0$.
	
	If $R(x_0,y_0)\le 1$, 
	Proposition~\ref{prop:peak}(i) and \eqref{eq:bf-sir-model} with $y_0>0$ imply that $R(x(t), y(t)) < 1$ for all $t \geq0$ and \eqref{y-max} implies that $y(t) \leq y_0 \leq \ov y$ for all $t \geq0$. Hence, also in this case the solution $(x(t), y(t))$ will remain in $\mc D_{\ov y}^-$ for all $t \geq 0$.
\end{proof}

The next result provides information about the dynamical behavior on the set $\mc D_{\ov y}^+$, and guarantees statement (ii). 
\begin{lemma}\label{lemma:y-hat}
	Consider the BF-SIR epidemic model \eqref{eq:bf-sir-model} with initial state $(x_0, y_0)$ in $\mc D_{\ov y}^{+}$. There exists a time instant $\hat t$ such that $y(\hat t) > \ov y$.
\end{lemma}
\begin{proof}
	From $(x_0, y_0)$ in $\mc D_{\ov y}^{+}$, it follows that $\phi^{-1} (\psi(x_0,y_0))> \ov y$ and $R(x_0,y_0)> 1$. Proposition~\ref{prop:peak}(ii) implies that there exists a time instant $\hat t$ such that $y(\hat t) $ is the maximum peak of infection. From Proposition~\ref{prop:y-max} and \eqref{y-max} it follows that 
	$$y(\hat t) = \phi^{-1} (\psi(x_0,y_0))> \ov y,$$
	thus proving the result.
\end{proof}\smallskip

We have shown that for initial conditions in $\mc D_{\ov y}^{+}$ an intervention is required to keep the infection curve below the threshold constraint. We discuss this in details in Section~\ref{sec:opt}. For that discussion, it will be useful to characterize the fraction of susceptibles for the dynamics \eqref{eq:bf-sir-model} with an initial state $(x_0, y_0)$ in $\mc D_{\ov y}^{+}$ when the fraction of infected reaches the threshold $\ov y$ for the first time. We define this value as $h(x_0, y_0)$. Note that such a point exists because $y(t)$ is continuous with initial state $y_0 \leq \ov y$ and, by Lemma~\ref{lemma:y-hat}, there exists a time instant $\hat t$ such that $y(\hat t) > \ov y$. We also denote $\ov x$ as the fraction of susceptibles such that $R(\ov x, \ov y) = 1$ and note that $\ov x \in (0,1-\ov y)$ because $R(1-\ov y, \ov y)>1$ and $R_x(x,y)>0$, by \eqref{partialR_x}.
Values of $h(x_0, y_0)$ and $\ov x$ can be visualized in Figure~\ref{fig:phase}, where the light-blue dashed line is a solution with initial condition $(x_0, y_0) \in \mc D_{\ov y}^+$.
\begin{lemma}\label{lemma:h}
	Let $h$ be as defined above. Then, for all $(x,y) \in \mc D_{\ov y}^{+}$,
	\begin{align}
		&h(x,y) \in (\ov x, 1-\ov y], \label{hyp-h-1}\\
		&h(x,y) \leq x,\label{hyp-h-2}\\
		&h(x, \ov y) = x.\label{hyp-h-3}
	\end{align}
	Moreover, $h$ is $\mc C^1$ on $\mc D_{\ov y}^{+}$ and
	\be\label{eq:hx} h_x(x,y) =  \frac{\rho(x,y)e^{g(h(x,y))}}{\rho(h(x,y),\ov y)e^{g(x)}}\,,\ee
	\be\label{eq:hy} h_y(x,y) = \frac{e^{g(h(x,y))}}{\rho(h(x,y),\ov y)e^{g(x)}}\,,\ee
	for every $(x,y)$ in $\mc D_{\ov y}^{+}$. 
\end{lemma}
\begin{proof} 
	Fix an initial state $(x_0,y_0) \in \mc D_{\ov y}^{+}$ for \eqref{eq:bf-sir-model}, by Lemma~\ref{lemma:y-hat}, there exists a time instant $\hat t$ such that $y(\hat t)$ is the peak of infection and $y(\hat t)> \ov y$. Hence, by Proposition~\ref{prop:peak}(ii), $R(x(t), y(t))>1$ for all $t \in [0,\hat t)$. By continuity, there exists a time instant $\ov t < \hat t$ such that $y(\ov t) = \ov y$ and $x(\ov t)= h(x_0,y_0)$. This means $R(x(\ov t), y(\ov t)) = R(h(x_0, y_0), \ov y)>1$. By \eqref{partialR_x} and $R(\ov x, \ov y)= 1$ it follows that $h(x_0,y_0)> \ov x$ and since $(h(x_0, y_0), \ov y) \in \mc S_+$, then $h(x_0,y_0) \leq 1-\ov y$. This proves \eqref{hyp-h-1}.
	
	From~Lemma \ref{prop:wellposedness}(i), $\dot{x}(t) \leq 0$ for all $t \geq 0$. Hence $x(\ov t)\!\leq \! x_0$, which means $h(x_0,y_0) \leq x_0$, proving \eqref{hyp-h-2}.
	
	Note that \eqref{hyp-h-3} follows directly from the definition of $h$.
	
	Define now the function $\theta :(\ov x,1-\ov y] \to \mathbb{R}$ as 
	\be \label{eq:theta}\theta(x) = \psi(x,\bar y)\,,\qquad\forall x\in(\ov x,1-\ov y]\,,\ee
	where $\psi$ is defined in \eqref{psi}. Since $\psi$ is $\mc C^1$ on $\mc S_+$, so is $\theta$ on $(\ov x,1-\ov y]$. 
	From the second equation in \eqref{eq:psi-partial} evaluated in $(x,\bar{y})$, we get
	\begin{align}
		\theta'(x) &= \psi_x(x,\bar y) =\rho(x, \bar{y})e^{-g(x)}> 0 \label{eq:theta-der}
	\end{align}
	for every $x$ in $(\ov x,1-\ov y]$. This follows from the fact that $R(x,\ov y)>1$ by $R(\ov x, \ov y)=1$ and \eqref{partialR_x}, which implies $\rho(x,\bar y)>0$. Therefore, there exists a continuous inverse function $\theta^{-1}: (\psi(\ov x,\ov y),\psi(1-\ov y,\ov y)] \to (\ov x,1-\ov y]$ that is $\mc C^1$ on $(\psi(\ov x,\ov y),\psi(1-\ov y,\ov y)]$, from the implicit function theorem~\cite[Theorem 7.1]{CanutoTabacco2008}. 
	Since $\psi(x,y)$ is constant along the orbits of~\eqref{eq:bf-sir-model} by Proposition~\ref{prop:invariant}, it follows that for all $(x,y)$ in $\mc D_{\ov y}^{+}$
	\be\label{eq:h-psi} \psi(x,y) = \psi(h(x,y), \ov y)= \theta(h(x,y))\,,\ee
	since $h(x,y) \in (\ov x, 1-\ov y]$ by \eqref{hyp-h-1}.
	Hence, for all $(x,y)$ in $\mc D_{\ov y}^{+}$
	\be\label{eq:h-def1}h(x,y) =  \theta^{-1}(\psi(x,y))\,.\ee
	
	
	Notice that $h$ is $\mc C^1$ on $\mc D_{\ov y}^{+}$ by \eqref{eq:h-def1}, since so are $\psi(x,y)$ on $\mc S_+$ and $\theta^{-1}$ on $(\psi(\ov x,\ov y),\psi(1-\ov y,\ov y)]$ and by $\psi(x,y) \in (\psi(\ov x,\ov y),\psi(1-\ov y,\ov y)]$. 
	Hence, computing the partial derivatives for both sides of \eqref{eq:h-def1} yields 
	\be \label{eq:partial_h} h_x(x,y) = \frac{\psi_x(x,y)}{\theta'(h(x,y))}\,,\quad h_y(x,y) = \frac{\psi_y(x,y)}{\theta'(h(x,y))}\,.\ee
	Substituting \eqref{eq:psi-partial} and \eqref{eq:theta-der} in \eqref{eq:partial_h} we get \eqref{eq:hx} and \eqref{eq:hy}. 
\end{proof}\smallskip

\section{Optimal control problem}\label{sec:opt}
Motivated by the previous analysis, we introduce an exogenous control $u: \R_+ \to [0,1]$, assumed to be a piecewise continuous function, which is determined by a social planner aiming to limit the spread of the disease by restricting social interactions.  
The controlled model is governed by the following system:
\be\label{control-system}\dot x=-(1-u)\frac{b(x)}{1+ a y }xy\,,\qquad\dot y=(1-u)\frac{b(x)}{1+ a y }xy-\gamma y\,.\ee

The following result establishes well-posedness of the initial value problems for the controlled BF-SIR epidemic model \eqref{control-system}.
\begin{lemma}\label{lemma:well-posedness-control} 
	For every initial state $(x_0, y_0)$  in $\mc S$ and every piecewise continuous control signal $u(t)$, the controlled BF-SIR epidemic model \eqref{control-system} admits a unique piecewise-$\mc C^1$ solution $(x(t), y(t)) \in \mc S$ for $t\ge0$.
\end{lemma}
\begin{proof}
	See Appendix. 
\end{proof}

The goal is to minimize the economic cost of the control while satisfying the constraint $y(t) \leq \ov y$ for all $t \geq 0$.
For a given initial condition $(x_0, y_0)$ in $\mc D_{\ov y}$, we consider the set $\mc U_{(x_0, y_0,\bar{y})}$ of piecewise continuous control signals $u: \mathbb{R}_+\to[0,1] $ such that the solution $(x(t),y(t))$ of the controlled BF-SIR epidemic model \eqref{control-system} with initial state $(x_0, y_0)$ and control $u$ satisfies 
\be\label{alwaysinD}(x(t), y(t)) \in \mc D_{\ov y}\,,\qquad \forall t \geq 0\,.\ee
Notice that the set  $\mc U_{(x_0, y_0,\bar{y})}$ is never empty since it contains the trivial control $u\equiv1$.
We model the cost through a functional $J : \mc U_{(x_0,y_0, \bar{y})} \to [0,+\infty]$ defined by
$$J(u) = \int_{0}^{+\infty} u(t) \, dt\,,$$
and consider the optimal control problem 
\be\label{control-problem} V^*(x_0,y_0)=\min_{u \in \mc U_{(x_0,y_0, \bar{y})}} J(u)\,.\ee

This type of optimal control problem has been widely studied in the literature \cite{Miclo.ea:2022, ACEMOGLU2024111391}. Within the framework of the classical epidemic SIR model, it has been shown that the optimal control is the so-called \emph{filling-the-box strategy}, which applies the minimum restrictions necessary to satisfy the threshold constraint. Inspired by such results, we consider the following control policy in feedback form: 
\be\label{feedback}u^*(t) = \mu(x^*(t), y^*(t)),\qquad \forall t\geq 0\,,\ee 
where $(x^*(t), y^*(t))$ is the solution of the controlled BF-SIR epidemic model \eqref{control-system} with initial state $(x_0, y_0)$ and control $u^*$, and $\mu:\mc S\to[0,1]$ is the feedback policy defined by 
\be\label{eq:mu}
\mu(x,y) = \begin{cases}
	0 & \text{ if } y < \bar{y} \text{ or } y = \ov y,\, x \leq \ov x\\
	\ds\rho(x,\ov y) & \text{ if } y = \bar{y},\,x > \ov x \,. 
\end{cases} \ee
Note that \eqref{feedback} guarantees $\dot{y}^*(t) \leq 0$ at any time instant $t$ such that $y^*(t) = \ov y$. 
\MA{Hence, it represents a feasible control. Assessing the optimality of \eqref{feedback} for the BF-SIR epidemic model is an open problem.}

Leveraging the results in Section~\ref{sec:threshold}, we characterize the cost for the policy \eqref{feedback}-\eqref{eq:mu} for the controlled BF-SIR epidemic model~\eqref{control-system}.
We write $J(u|x_0, y_0)$ to stress the role of the initial condition.
\begin{proposition}\label{prop:Ju*}
	For every $(x,y)$ in $\mc D_{\ov y}$, 
	\be\label{Ju*} \!\! J(u^*| x_0 ,y_0) =  \begin{cases}
		0 \!\!& \text{if }  (x_0,y_0) \in \mc D_{\ov y}^{-} \\
		\displaystyle	\frac1{\gamma \bar{y}}\int_{\bar{x}}^{h(x_0,y_0)}  \!\! \rho(x,\bar y) \de x	 \!\!& \text{if }   (x_0,y_0) \in \mc D_{\ov y}^{+}.
	\end{cases}\ee
\end{proposition}\medskip
\begin{proof}
	Fix $(x_0,y_0)$ in $\mc D_{\ov y}$, and let $(x^*(t),y^*(t))$, for $t\ge0$, be the solution of controlled behavioral-feedback SIR epidemic model \eqref{control-system} with initial state $(x_0, y_0)$ and feedback control \eqref{feedback}--\eqref{eq:mu}. We next prove that if $(x,y) \in \mc D_{\ov y}^-$, then $\mu(x,y) =0$. This is clearly true for $y < \ov y$. For $y = \ov y$, we have that $(x,\ov y)$ belonging to $\mc D_{\ov y}^-$ implies $R(x,\ov y)<1$.
	From $R(\ov x, \ov y)=1$ and \eqref{partialR_x}, we have that $R(x, \ov y) <1$ implies $x < \ov x$. Hence, $\mu(x, \ov y)=0$.
	By Lemma \ref{lemma:invariant-D-}, if $(x_0,y_0)\in\mc D_{\ov y}^-$, then $(x^*(t),y^*(t))\in\mc D_{\ov y}^-$. This proves that $J(u^*|x_0,y_0) = 0$ for every $(x_0,y_0)$ in $\mc D_{\bar y}^-$.
	On the other hand, if $(x_0,y_0)$ in $\mc D_{\ov y}^+$, let $t^*>0$ be the first time at which $y^*(t^*) = \ov y$. Recall that the corresponding fraction of susceptibles is $x^*(t^*) = h(x_0, y_0) > \ov x$, by Lemma \ref{lemma:h}. Under the control $u^*$, the cost is $0$ for $t \in [0, t^*]$.
	At $t^*$, the dynamics becomes 
	$$\dot x^*=-(1-u^*)\frac{b(x)}{1+ a \ov y }x\ov y=-\gamma\ov y,\quad \dot y^* = 0\,,$$
	hence keeping $y^*(t)=\ov y$, until reaching the point $(x^*(\ov t),y^*(\ov t))=(\ov x,\ov y)$ at some time $\ov t$. 
	Since $R(\ov x, \ov y)=1$, by Proposition \eqref{prop:peak}(i), $y(t) < \ov y$ and $u^*(t)=0$ with zero cost for all $t \geq \ov t$.
	The total cost between the two time instants $t^*$ and $\ov t$ is given by 
	\begin{align*}
		\int_{t^*}^{\ov t}u^*(t)\de t&=\int_{t^*}^{\ov t}\rho(x^*(t),y^*(t))\de t \\
		&=\int_{t^*}^{\ov t}\rho(x^*(t),\ov y)\de t\\
		&=\frac1{\gamma \bar{y}}\int_{\bar{x}}^{h(x_0,y_0)}\rho(x,\ov y)\de x\,.
	\end{align*}
\end{proof} \smallskip

The next result characterizes the sensitivity of the cost of $u^*$ with respect to the initial condition. 
\begin{proposition}\label{prop:sensitivity}
	Consider the cost \eqref{Ju*}. Then, for all $(x,y) \in\mc D_{\ov y}^{+}$,
	\begin{align}
		\ds J_{y}(u^*|x,y)&=\frac{1}{\gamma\ov y}e^{\int_x^{h(x,y)}\frac{\gamma a}{sb(s)}\de s}\,,\label{eq:partial-jy}\\
		\ds J_{x}(u^*|x,y)& =\rho(x,y)J_{y}(u^*,x,y) \label{eq:partial-jx}\,.
	\end{align}
\end{proposition}\medskip
\begin{proof}
	From a direct differentiation of the expression of $J$ on $\mc D_{\ov y}^{+}$ and \eqref{eq:hy}, \eqref{eq:hx}, \eqref{eq:g} we get
	\begin{align*}
		J_{y}(u^*|x,y) &= \frac{\rho(h(x,y),\ov y)}{\gamma \bar{y}}h_{y}(x,y) \\
		&= \frac{1}{\gamma \bar{y}}e^{g(h(x,y))-g(x)}= \frac{1}{\gamma\ov y}e^{\int_x^{h(x,y)}\gamma a/(sb(s))\de s} ,
	\end{align*}
	and 
	\begin{align*}
		J_{x}(u^*|x,y) & = \frac{\rho(h(x,y),\ov y)}{\gamma \bar{y}}h_{x}(x,y)  \\
		& =  \frac{\rho(x,\ov y)}{\gamma \bar{y}}e^{g(h(x,y))-g(x)}= \rho(x,y) J_{y}(u^*|x,y).
	\end{align*}
\end{proof}

From \eqref{eq:partial-jy}, it follows that $J_{y}(u^*|x,y) > 0$ for all $(x,y) \in \mc D_{\ov y}^{+}$. Note also that, since $h(x,y) \leq x$ by Lemma \ref{lemma:h} and $g(x)$ is increasing in $x$, then $g(h(x,y)) \leq g(x)$ and 
\be \label{eq:jy-1} J_{y}(u^*|x,y) ) \leq \frac{1}{\gamma \ov y}\,,\ee
for all $(x,y) \in \mc D_{\ov y}^{+}$, and, since $h(x, \ov y) = x$ by Lemma \ref{lemma:h}, 
\be \label{eq:jy-2}J_{y}(u^*|x,\ov y) = \frac{1}{\gamma \ov y}\,,\ee
for all $(x,y) \in \mc D_{\ov y}^{+}$ with $y = \ov y$.
Overall, Propositions \ref{prop:Ju*} and \ref{prop:sensitivity} prove that $J(u^*|x,y)$ is differentiable for all $(x,y) \in \mc S \setminus\{\Gamma \cup \{y=0\}\}$. In this set, the equations above imply that $J(u^*|x,y)$ satisfies the Hamilton-Jacobi-Bellman equation \cite{liberzon}, that is, the minimum of 
\begin{align*}J_x(u^*| x,y)\dot{x} + &J_{y}(u^*|x,y)  \dot{y} +u = \\
	& = (\rho(x,y)\dot{x} + \dot{y} )J_{y}(u^*|x,y) )  + u \\
	&= (1-  \gamma y J_{y}(u^*|x,y) )) u 
\end{align*}
is $0$, by \eqref{eq:jy-1}, and $u^*$ is the minimizer. The last fact follows since either $y < \ov y$, which implies by \eqref{feedback}--\eqref{eq:mu} that $u^* = 0$, or $y = \ov y$, which implies by \eqref{eq:jy-2} that $1-\gamma y J_y(u^*|x,y) = 0$. 
This argument cannot be extended to the entire domain $\mc D_{\ov y}$ since one can prove that $J$ is not differentiable in $\Gamma$. Hence, the proof of optimality of $u^*$ requires a more detailed analysis and is left as future work. 

	\section{Conclusion}\label{sec:conclusion}
	This paper presents a behavioral-feedback SIR model in which the infection rate adapts dynamically based on the current state of the epidemic. 
	For this model, we derive an invariant of motion which allows a precise characterization of the maximum infection peak and of the dynamical behavior of the solution under a threshold constraint. In addition, we formulate an optimal control problem to balance intervention costs and public health goals and characterize the cost associated with a control strategy that enforces the minimum restrictions necessary to avoid exceeding the health capacity. 
	
	Future research could generalize the model with a broader dependence of the transmission rate on the fraction of infected individuals and further investigate the optimality of the filling-the-box strategy.
	
	\bibliographystyle{IEEEtran}
	\bibliography{bib}
	
	\appendix

	\section{Proof of Lemma \ref{lemma:well-posedness-control}}\label{sec:proof-well-posedness-control}
	The right-hand side of \eqref{control-system} is locally Lipschitz with respect to $(x, y, z)$. On any time interval on which $u$ is continuous, local existence and uniqueness of a solution for the corresponding Cauchy problem is a consequence of the Picard-Lindelöf theorem \cite[I.3]{hale}. 
	If $u$ admits jump discontinuities at times $0<t_1<t_2<\cdots$, we define $u^k:[t_k, t_{k+1}]\to\R$ for $k\geq 0$ (with $t_0=0$) as the unique continuous function that coincides with $u$ on $[t_k, t_{k+1}]$. Given an initial condition $(x_0,y_0)$, we iteratively solve the Cauchy problem relative to \eqref{control-system} with such initial condition as follows. On every interval $[t_k, t_{k+1}]$ we use the control $u^k$ and solve for initial conditions $(x_k,y_k)$. We call $(x^k(t), y^k(t))$ the corresponding solution. We then recursively impose that $(x_{k+1},y_{k+1})=(x^k(t_{k+1}), y^k(t_{k+1}))$. This yields, by construction, a piecewise $\mathcal{C}^1$ maximal solution $(x(t), y(t))$ that coincides with $(x^k(t), y^k(t))$ on $[t_k, t_{k+1}]$. 
	
	Integrating \eqref{control-system}, assuming an initial condition $(x_0, y_0) \in \mc S$, yields
	\begin{align}
		x(t) &= x_0 \exp\left( \int_{0}^{t} \left( -u(\tau)\frac{b(x(\tau))}{1+ a y(\tau)}   y(\tau) \right)d\tau \right), \\
		y(t) &= y_0 \exp\left( \int_{0}^{t} \left( u(\tau)\frac{b(x(\tau))}{1+ a y(\tau)}   x(\tau)- \gamma \right)d\tau \right).
	\end{align}
	Moreover, summing the equations in \eqref{control-system} gives $\dot{x}(t) + \dot{y}(t) = - \gamma y$.
	Overall, as long as the solution $(x(t), y(t))$ exists, we must have that 
	$$x(t) \geq  0\,, \quad y(t) \geq 0\,,  \quad x(t) + y(t)\leq 1 \,.$$
	In other terms, as long as the solution exists, it lives in $\mc S$.  In particular, this implies that the solution is globally defined. 
	
\end{document}